\documentclass[11pt,a4paper]{article}

\usepackage{amsmath}
\usepackage{amsthm}
\usepackage{amssymb}
\usepackage{cain_arxiv}
\usepackage{graphicx}
\usepackage{booktabs}

\theoremstyle{definition}
\newtheorem{definition}{Definition}[section]
\newtheorem{example}[definition]{Example}

\theoremstyle{plain}
\newtheorem{proposition}[definition]{Proposition}
\newtheorem{lemma}[definition]{Lemma}

\numberwithin{equation}{section}

\def\fullref#1#2{%
  \ifdefined\hyperref%
    {\hyperref[#2]{#1\space\penalty 200\relax\ref*{#2}}}%
  \else%
    {#1\space\penalty 200\relax\ref{#2}}%
  \fi%
}

\newcommand{\defterm}[1]{\textit{#1}}

\newcommand{\pres}[2]{\left\langle #1\:|\:#2 \right\rangle}

\newcommand{\nset}{\mathbb{N}}

\newcommand{\emptyword}{\varepsilon}
\newcommand{\rel}[1]{\mathcal{#1}}

\newcommand{\thue}{\leftrightarrow^*}
\newcommand{\imthue}{\leftrightarrow}
\newcommand{\imred}{\rightarrow}
\newcommand{\red}{\rightarrow^*}

\begin{document}

\title{Finitely presented monoids with linear Dehn function need not have regular cross-sections}
\author{Alan J. Cain \& Victor Maltcev}
\date{}

\thanks{During the research that led to the this paper, the
  first author was supported by the European
  Regional Development Fund through the programme {\sc COMPETE} and by
  the Portuguese Government through the {\sc FCT} (Funda\c{c}\~{a}o
  para a Ci\^{e}ncia e a Tecnologia) under the project
  {\sc PEst-C}/{\sc MAT}/{\sc UI}0144/2011 and through an {\sc FCT} Ci\^{e}ncia 2008
  fellowship.}

\maketitle

\address[AJC]{%
Centro de Matem\'{a}tica, Universidade do Porto, \\
Rua do Campo Alegre 687, 4169--007 Porto, Portugal
}
\email{%
ajcain@fc.up.pt
}
\webpage{%
www.fc.up.pt/pessoas/ajcain/
}
\address[VM]{Rudka 24, Khvoshchove, Shyshaky Subregion, Poltava Region, Ukraine}
\email{%
victor.maltcev@gmail.com
}

\begin{abstract}
This paper shows that a finitely presented monoid with linear Dehn
function need not have a regular cross-section, strengthening the
previously-known result that such a monoid need not be presented by a
finite complete string rewriting system, and contrasting the fact that
finitely presented groups with linear Dehn function always have
regular cross-sections.

\end{abstract}

\section{Introduction}

The use of isoperimetric and Dehn functions in group theory stems from
the seminal paper of Gromov~\cite{gromov_hyperbolic} and its
characterization of word-hyperbolic groups as groups having linear Dehn
function. Another characterization of word-hyperbolic groups is admitting a
Dehn presentation (equivalently, a presentation via a finite complete
length-reducing rewriting system). Since the language of irreducible
words of such a presentations is regular, this shows that any group
with linear Dehn function has a regular cross-section.

Indeed, the word problem for both groups and monoids is closely tied
to Dehn functions and cross-sections. It is well-known that a finitely
presented monoid (possibly a group) has soluble word problem if and
only if it had recursive Dehn function if and only if it has a
recursive cross-section. Squier \cite{squier_wordproblem} gave
examples of monoids with soluble word problem that cannot be presented
by finite complete rewriting systems, but these monoids still have
regular cross-sections. Kobayashi \cite[\S~4]{kobayashi_noregfcrs}
gave an example of a finitely presented monoid with soluble word
problem that does not admit a regular complete presentation; Otto,
Katsura \& Kobayashi \cite[Example~6.4]{otto_infiniteconvergent}
proved the stronger result that this monoid does not have a regular
cross-section. However, it is easy to see that this example has
quadratic Dehn function. Otto, Sattler-Klein \& Madlener
\cite[\S~3]{otto_automonversus} gave an example of a linear Dehn
function monoid that cannot be presented by a finite complete
rewriting system. However, this example still has a regular
cross-section. In the present paper, \fullref{Example}{ex:ldf}
exhibits a linear Dehn function monoid that does not have a regular
cross-section, thus filling in the last line in
\fullref{Table}{tbl:summary}.

\begin{table}[t]
\caption{Summary of known examples relating a monoid's Dehn function
  and whether it has a presentation via a finite complete rewriting
  system or a regular cross-section. Notice that if a monoid does not
  have a regular cross-section, it cannot have a presentation via a
  finite complete rewriting system.}
\label{tbl:summary}
\begin{center}
\begin{tabular}{cccc}
\toprule
& Dehn function & FCRS presentation & Regular cross-section \\
\midrule
\cite{squier_wordproblem} & Recursive & No & Yes \\
\cite{otto_infiniteconvergent} & Quadratic & No & No \\
\cite{otto_automonversus} & Linear & No & Yes \\
\fullref{Example}{ex:ldf} & Linear & No & No \\ 
\bottomrule
\end{tabular}
\end{center}
\end{table}

But the real importance of this example comes from the r\^{o}le it
plays in the possible generalization to monoids of word-hyperbolicity,
using the various equivalent characterizations of word-hyperbolic
groups. For instance, one may consider finitely generated monoids
whose Cayley graphs of which form hyperbolic spaces (see, for example,
\cite{c_rshcg,cassaigne_infinitewords}); or one could monoids that
satisfy Gilman's linguistic characterization of word-hyperbolicity for
groups \cite{gilman_hyperbolic} (see
\cite{cm_wordhypunique,duncan_hyperbolic}). Another possibility would
be to consider monoids having linear Dehn function. However, the
example in this paper shows that if one attempts to generalize using
linear Dehn functions, one must abandon one of the fundamental
properties of word-hyperbolic groups, viz.,~having a regular
cross-section.

[The research described in this paper has been included in the second
  author's Ph.D.~thesis \cite[\S~7.2]{maltcev_phd}.]

\section{Preliminaries}

We briefly recall the necessary definitions and terminology; see
\cite{book_srs} for further background information on for
string-rewriting systems.

A \defterm{string rewriting system}, or simply a \defterm{rewriting
  system}, is a pair $(A,\rel{R})$, where $A$ is a finite alphabet and
$\rel{R}$ is a set of pairs $(l,r)$, usually written $l \imred r$,
known as \defterm{rewriting rules} or simply \defterm{rules}, drawn
from $A^* \times A^*$. The single reduction relation
$\imred_{\rel{R}}$ is defined as follows: $u \imred_{\rel{R}} v$
(where $u,v \in A^*$) if there exists a rewriting rule $(l,r) \in
\rel{R}$ and words $x,y \in A^*$ such that $u = xly$ and $v =
xry$. That is, $u \imred_{\rel{R}} v$ if one can obtain $v$ from $u$
by substituting the word $r$ for a subword $l$ of $u$, where $l \imred r$
is a rewriting rule. The reduction relation $\red_{\rel{R}}$ is the
reflexive and transitive closure of $\imred_{\rel{R}}$. The process of
replacing a subword $l$ by a word $r$, where $l \imred r$ is a rule,
is called \defterm{reduction} by application of the rule $l \imred r$;
the iteration of this process is also called reduction. A word $w \in
A^*$ is \defterm{reducible} if it contains a subword $l$ that forms
the left-hand side of a rewriting rule in $\rel{R}$; it is otherwise
called \defterm{irreducible}.

The rewriting system $(A,\rel{R})$ is \defterm{finite} if both $A$ and
$\rel{R}$ are finite. The rewriting system $(A,\rel{R})$ is
\defterm{noetherian} if there is no infinite sequence $u_1,u_2,\ldots
\in A^*$ such that $u_i \imred_{\rel{R}} u_{i+1}$ for all $i \in
\nset$. That is, $(A,\rel{R})$ is noetherian if any process of
reduction must eventually terminate with an irreducible word. The
rewriting system $(A,\rel{R})$ is \defterm{confluent} if, for any
words $u, u',u'' \in A^*$ with $u \red_{\rel{R}} u'$ and $u
\red_{\rel{R}} u''$, there exists a word $v \in A^*$ such that $u'
\red_{\rel{R}} v$ and $u'' \red_{\rel{R}} v$. The rewriting system
$(A,\rel{R})$ is \defterm{locally confluent} if, for any words $u,
u',u'' \in A^*$ with $u \imred_{\rel{R}} u'$ and $u \imred_{\rel{R}}
u''$, there exists a word $v \in A^*$ such that $u' \red_{\rel{R}} v$
and $u'' \red_{\rel{R}} v$. A \defterm{critical pair} occurs when of
left-hand sides of two (not necessarily distinct) rewriting rules can
overlap: two rules $l_1 \imred r_1)$ and $l_2 \imred r_2$, such that
there is either (1) $l_1 = xy$ and $l_2 = yz$, so that $xyz
\imred_{\rel{R}} r_1z$ and $xyz \imred_{\rel{R}} xr_2$, or (2) $l_1 =
xl_2y$, so that $l_1 \imred_{\rel{R}} r_1$ and $l_1 = xl_2y
\imred_{\rel{R}} xr_2y$. If, in case (1), there is a word $w$ such
that $r_1z \red_{\rel{R}} w$ and $xr_2 \red_{\rel{R}}$, or, in case
(2), there is a word $w$ such that $r_1 \red_{\rel{R}} w$ and $xr_2y
\red_{\rel{R}} w$, then the critical pair is said to
\defterm{resolve}. A rewriting system is locally confluent if and only
if every critical pair resolves. A noetherian rewriting system is
confluent if and only if it is locally confluent. A rewriting system
that is both confluent and noetherian is \defterm{complete}.

The \defterm{Thue congruence} $\thue_{\rel{R}}$ is the equivalence
relation generated by $\imred_{\rel{R}}$. The elements of the monoid
presented by $\pres{A}{\rel{R}}$ are the $\thue_{\rel{R}}$-equivalence
classes. If $(A,\rel{R})$ is complete, every
$\thue_{\rel{R}}$-equivalence class contains a unique irreducible
word, and any word in this class must reduce to that unique
irreducible word. Two rewriting systems $(A,\rel{R})$ and
$(A,\rel{S})$ are \defterm{equivalent} if $\thue_{\rel{R}}$ and
$\thue_{\rel{S}}$ coincide, in which case the monoids presented by
$\pres{A}{\rel{R}}$ and $\pres{A}{\rel{S}}$ are isomorphic.

A \defterm{cross-section} for the monoid presented by
$\pres{A}{\rel{R}}$ is a language $L$ over $A$ containing exactly one
word in each $\thue_{\rel{R}}$-equivalence class. The set of
irreducible words of a complete rewriting system $(A,\rel{R})$ forms a
cross-section for the monoid presented by $\pres{A}{\rel{R}}$. If
$(A,\rel{R})$ is also finite, the set of irreducible words,
\[
A^* - A^*\{l : (l \imred r) \in \rel{R}\}A^*,
\]
is a regular cross-section of the monoid.

\begin{proposition}[{\cite[Proposition~5.3]{cm_markov}}]
\label{prop:changegen}
Let $M$ be a monoid that has a regular cross-section over some
finite generating set. Then for every finite generating set $A$ for
$M$, there is a regular cross-section of $M$ over $A$.
\end{proposition}

Let $(A,\rel{R})$ be a finite rewriting system and let $M$ be the
monoid presented by $\pres{A}{\rel{R}}$. The \defterm{Dehn function}
of this presentation, $D_{M;A,\rel{R}} : \nset \to \nset$, is defined
as follows. For two words $u,v\in A^*$ such that $u \thue_{\rel{R}}
v$, let $d_{\rel{R}}(u,v)$ be the least number of relations from $\rel{R}$
that can be applied to obtain $v$ from $u$. Then
\[
\mathrm{D}_{M;A,\rel{R}}(n)=\max\Bigl\{d_{\rel{R}}(u,v) : u,v\in A^*,
|u|+|v|\leq n, u \thue_{\rel{R}} v\Bigr\}.
\]
It is easy to see that the growth rate (in terms of $n$) of the Dehn
function depends only on $M$, not on the choice of the presentation
$\pres{A}{\rel{R}}$. Thus this growth rate is an invariant of $M$. The
monoid $M$ has \defterm{linear Dehn function} if
$\mathrm{D}_{M;A,\rel{R}}(n)$ grows linearly.

Note in passing that $d_{\rel{R}}(puq,pvq) \leq d_{\rel{R}}(u,v)$ for
all $u,v,p,q \in A^*$ with $u \thue_{\rel{R}} v$.

\section{The example}

This paper is centred on the study of the following monoid:

\begin{example}
\label{ex:ldf}
Let $A = \{a,b,c,0\}$ and let $\rel{R}$ consist of the following
rewriting rules:
\begin{align}
ba &\imred a^2b, \tag{BA}\label{eq:ba} \\
bc &\imred aca, \tag{BC}\label{eq:bc} \\
ac^2 &\imred 0, \tag{ACC}\label{eq:acc} \\
x0 &\imred 0,\quad 0x \imred 0 \qquad \tag{Z}\text{for all $x \in A$}. \label{eq:z}
\end{align}
Notice that $(A,\rel{R})$ is a finite rewriting system.

Let $M$ be the monoid presented by $\pres{A}{\rel{R}}$. By
\fullref{Proposition}{prop:ldf} and~\ref{prop:noregcs} below, $M$ has
linear Dehn function and does not have a regular cross-section. In
particular, $M$ cannot be presented by a finite complete rewriting
system.
\end{example}

[The rewriting system $(A,\rel{R})$ has a \textit{prima facie}
  resemblance to the example proven by Otto, Katsura \&
  Kobayashi \cite[Example~6.4]{otto_infiniteconvergent} not to have a
  regular cross-section; however, as observed in the introduction,
  this example has quadratic Dehn function.]

We begin with some preliminary results before proceeding to show that
$M$ has linear Dehn function and does not have a regular
cross-section.

Let $\rel{S}$ consist of the rewriting rules \eqref{eq:ba}, \eqref{eq:bc},
and \eqref{eq:z}, and also
\begin{align}
a^{2^{n+1}-1}ca^nc &\imred 0 \qquad \text{for all $n \geq 0$}, \tag{ACAC}\label{eq:acac} 
\end{align}
Notice that the rewriting system $(A,\rel{S})$ is infinite, and
further that \eqref{eq:acc} is simply \eqref{eq:acac} with $n = 0$.

\begin{lemma}
\label{lem:equivalent}
The rewriting systems $(A,\rel{R})$ and $(A,\rel{S})$ are equivalent
(that is, $\thue_{\rel{R}}$ and $\thue_{\rel{S}}$ coincide) and hence
$\pres{A}{\rel{S}}$ also presents the monoid $M$.
\end{lemma}

\begin{proof}
Note that $\rel{R} \subseteq \rel{S}$. The only rules in $\rel{S}$
that are not in $\rel{R}$ are rules \eqref{eq:acac} for $n \geq
1$. Thus it suffices to show that
\begin{equation}
a^{2^{n+1}-1}ca^nc \thue_{\rel{R}} 0\label{eq:acacthues}
\end{equation}
for all $n \geq 1$. Proceed by induction on $n$: for $n=0$,
$a^{2^{n+1}-1}ca^nc = ac^2 \imthue_{\rel{R}} 0$. So suppose
\eqref{eq:acacthues} holds for $n = k-1$; the aim is to show
\eqref{eq:acacthues} holds for $n=k$:
\begin{align*}
& a^{2^{k+1}-1}ca^kc \\
\imthue_{\rel{R}}{}& a^{2^{k+1}-2}bca^{k-1}c && \text{(by \eqref{eq:bc})} \\
\thue_{\rel{R}}{}& ba^{2^{k}-1}ca^{k-1}c && \text{(by \eqref{eq:ba} applied $2^{k}-1$ times)} \\
\thue_{\rel{R}}{}& b0 && \text{(by induction; \eqref{eq:acacthues} with $n = k-1$)} \\
\imthue_{\rel{R}}{}& 0 && \text{(by \eqref{eq:z}).} 
\end{align*}
Hence, by induction, \eqref{eq:acacthues} holds for all $n \geq 1$.
\end{proof}

\begin{lemma}
\label{lem:complete}
The rewriting system $(A,\rel{S})$ is complete.
\end{lemma}

\begin{proof}
To see that the rewriting system $(A,\rel{S})$ is noetherian, notice
that rewriting rules \eqref{eq:z} and \eqref{eq:acac} reduce the
length of a word and \eqref{eq:bc} decreases the number of letters $b$
in a word. So any infinite sequence of reduction would involve
infinitely many applications of \eqref{eq:ba}. To see that this is
impossible, let $(a^{d_{k+1}}ba^{d_k}\cdots ba^{d_1})\theta =
(d_{k+1},d_k,\ldots,d_1)$ and define an ordering on $k$-tuples of
non-negative integers by:
\[
(d_{k+1},d_k,\ldots,d_1) < (e_{k+1},ba^{e_k},\ldots,e_1) \iff (\exists
h)\bigl((\forall i < h)(d_i = e_i) \land d_h < e_h\bigr).
\]
This is a well-ordering of such $k$-tuples. Since rewriting using
\eqref{eq:ba} always reduces (with respect to this order) the image
under $\theta$ of some maximal subword over $\{a,b\}^*$, it follows
that no sequence of reduction can involve infinitely many applications
of \eqref{eq:ba}.

To see that $(A,\rel{S})$ is confluent, notice first that any word
containing $0$ will always be rewritten to $0$ using rules
\eqref{eq:z}, and so any critical pairs involving any of these rules
always resolve. There are two remaining cases to consider.

First, the left-hand sides of \eqref{eq:ba} and some rule
\eqref{eq:acac} may overlap in a word $ba^{2^{k+1}-1}ca^kc$:
\[
ba^{2^{k+1}-1}ca^kc \imred_{\rel{S}} b0\quad\text{ and }\quad ba^{2^{k+1}-1}ca^kc \imred_{\rel{S}} a^2ba^{2^{k+1}-2}ca^kc.
\]
But $b0 \imred 0$ by \eqref{eq:z}, and
\begin{align*}
& a^2ba^{2^{k+1}-2}ca^kc \\
\red_{\rel{S}}{}& a^{2(2^{k+1}-1)}bca^kc && \text{(by \eqref{eq:ba} applied $2^{k+1} - 2$ times)} \\
\imred_{\rel{S}}{}& a^{2^{k+2}-1}ca^{k+1}c && \text{(by \eqref{eq:bc})} \\
\imred_{\rel{S}}{}& 0 && \text{(by \eqref{eq:acac} with $n=k+1$).}
\end{align*}
Thus such critical pairs resolve.

Second, two left-hand sides of rules \eqref{eq:acac} may overlap in a
word of the form $a^{2^{m + 2^{k+1}}-1}ca^{m+2^{k+1}-1}ca^{k}c$, where
$m \geq 0$:
\begin{align*}
a^{2^{m + 2^{k+1}}-1}ca^{m+2^{k+1}-1}ca^{k}c &\imred_{\rel{S}} 0a^kc,\\
a^{2^{m + 2^{k+1}}-1}ca^{m+2^{k+1}-1}ca^{k}c &\imred_{\rel{S}} a^{2^{m + 2^{k+1}}-1}ca^{m}0;
\end{align*}
but such a critical pair resolves because both of these words reduce to $0$.

Hence all critical pairs resolve and so $(A,\rel{S})$ is locally
confluent and hence, since it is noetherian, confluent. So
$(A,\rel{S})$ is complete.
\end{proof}

\begin{lemma}
\label{lem:distz}
For any word $w \in A^*$ containing a symbol $0$, $d_{\rel{R}}(w,0) \leq |w|$.
\end{lemma}

\begin{proof}
At most $|w|$ applications of rules \eqref{eq:z} reduces $w$ to $0$.
\end{proof}

\begin{lemma}
\label{lem:distacac}
For every $k \geq 0$,
\[
d_{\rel{R}}\bigl(a^{2^{k+1}-1}ca^kc,0\bigr) \leq 2^{k+1}+k-1,
\]
which is less than the length of the word $a^{2^{k+1}-1}ca^kc$.
\end{lemma}

\begin{proof}
Proceed by induction on $k$. For $k = 0$, $a^{2^{0+1}-1}aca^0c = ac^2
\imred_{\rel{S}} 0$, so $d_{\rel{R}}(a^{2^{0+1}-1}aca^0c,0) = 1 = 2^{0+1}+0-1$.

Now suppose the result holds for $k-1$; the aim is to show it
holds for $k$. Then:
\begin{align*}
d_{\rel{R}}\bigl(a^{2^{k+1}-1}ca^kc,a^{2^{k+1}-2}bca^{k-1}c\bigr) &\leq 1 && \text{(by \eqref{eq:bc})} \\
d_{\rel{R}}\bigl(a^{2^{k+1}-2}bca^{k-1}c,ba^{2^{k}-1}ca^{k-1}c\bigr) &\leq 2^k-1 && \text{(by \eqref{eq:ba})} \\
d_{\rel{R}}\bigl(ba^{2^{k}-1}ca^{k-1}c,b0\bigr) &\leq 2^k +k-2 && \text{(by induction)} \\
d_{\rel{R}}(b0,0) &\leq 1 && \text{(by \eqref{eq:z}).}
\end{align*}
Hence
\[
d_{\rel{R}}\Bigl(a^{2^{k+1}-1}ca^kc,0\Bigr) \leq 1 + (2^k - 1) + (2^{k} + k - 2) +1 = 2^{k+1} -1 + k.
\]
Therefore, by induction, the result holds for all $k$.
\end{proof}

Notice that a word $ba^{d_k}\cdots ba^{d_1}c$ (where $k,d_i \in \nset
\cup \{0\}$) can be reduced by iteratively moving the rightmost letter
$b$ to the right using rule \eqref{eq:ba} and then, when it is next to
the letter $c$, removing it by rule \eqref{eq:bc}. Each application of
\eqref{eq:bc} produces a single letter $a$ to the right of the letter
$c$. That is,
\[
ba^{d_k}\cdots ba^{d_1}c \red_{\rel{S}} a^{f(d_k,\ldots,d_1)}ca^k,
\]
for some function $f$. [Notice that $f$ is well-defined because
  $(A,\rel{S})$ is confluent and $a^{f(d_k,\ldots,d_1)}ca^k$ is
  irreducible with respect to $\rel{S}$.]

\begin{lemma}
\label{lem:f}
The function $f$ is defined on tuples of non-negative integers by
\[
f(d_k,\ldots,d_1)=2d_{k}+2^2d_{k-1}+\ldots+2^{k}d_1+2^{k}-1.
\]
\end{lemma}

\begin{proof}
First, we derive a recursive expression for $f(d_k,\ldots,d_1)$. Note
that $ba^{d_1}c \red_{\rel{S}} a^{2d_1}bc \imred_{\rel{S}}
a^{2d_1+1}ca$, so $f(d_1) = 2d_1+1$. By induction on $k$,
\begin{align*}
&ba^{d_k}ba^{d_{k-1}}\cdots ba^{d_1}c \\
\red_{\rel{S}}{}& ba^{d_k+f(d_{k-1},\ldots,d_1)}ca^{k-1} && \text{(by the definition of $f$)}\\
\red_{\rel{S}}{}& a^{2d_k+2f(d_{k-1},\ldots,d_1)}bca^{k-1} && \text{(by \eqref{eq:ba} applied $d_k+f(d_{k-1},\ldots,d_1)$ times)}\\
\imred_{\rel{S}}{}& a^{2d_k+2f(d_{k-1},\ldots,d_1)+1}ca^{k} &&\text{(by \eqref{eq:bc});} 
\end{align*}
thus $f(d_k,\ldots,d_1)=2f(d_{k-1},\ldots,d_1)+2d_k+1$. Hence, by induction on $k$,
\begin{align*}
f(d_k,\ldots,d_1) &= 2f(d_{k-1},\ldots,d_1)+2d_k+1 \\
&= 2(2d_{k-1}+2^2d_{k-2}+\ldots+2^{k-1}d_1+2^{k-1}-1) + 2d_k+1 \\
&= 2^2d_{k-1}+2^2d_{k-2}+\ldots+2^{k}d_1+2^{k}-2 + 2d_k+1 \\
&= 2d_k+2^2d_{k-1}+2^2d_{k-2}+\ldots+2^{k}d_1+2^{k}-1. \qedhere
\end{align*}
\end{proof}

\begin{proposition}
\label{prop:ldf}
The monoid $M$ has linear Dehn function.
\end{proposition}

\begin{proof}
It is necessary to show that there is a constant $C$ such that if $u,v
\in A^*$ are such that $u \thue_{\rel{R}} v$, then $d_{\rel{R}}(u,v) \leq
C(|u|+|v|)$. We consider two cases separately: $u \thue_{\rel{R}} v
\thue_{\rel{R}} 0$ and $u \thue_{\rel{R}} v \not\thue_{\rel{R}} 0$.

\medskip
\noindent\textit{First case.} Suppose $u \thue_{\rel{R}} v
\thue_{\rel{R}} 0$. In this case, it suffices to prove that there is a
constant $C$ such that $d_{\rel{R}}(w,0) < C|w|$ whenever $w \thue_{\rel{R}} 0$, for then
$d_{\rel{R}}(u,v) \leq d_{\rel{R}}(u,0) + d_{\rel{R}}(0,v) \leq C|u| + C|v| = C(|u| + |v|)$.

So let $w\in\{a,b,c,0\}^*$ be such that $w \thue_{\rel{R}} 0$ in
$M$. If $0$ is present in $w$, then $d_{\rel{R}}(w,0) \leq |w|$ by
\fullref{Lemma}{lem:distz}. So assume without loss of generality that
$w\in\{a,b,c\}^*$.

Apply the reverse of rule \eqref{eq:ba} to $w$ as much as possible,
replacing all subwords $a^2b$ with $ba$, and resulting in a word
$w'$. Since each such application decreases the number of symbols $a$,
it can be applied at most $|w|$ times, so $d_{\rel{R}}(w,w') \leq
|w|$. Furthermore, this process always results in a shorter word, so
$|w'| \leq |w|$.

Now, $w' \red_{\rel{S}} 0$. Let us reduce $w'$ as follows: at each
step, proceed as follows:
\begin{enumerate}

\item If there is a subword of the form $a^{2^{k+1}-1}ca^kc$, apply
  the rule \eqref{eq:acac} with $n=k$, to get a word of the form
  containing a symbol $0$, then reduce to $0$ using rules \eqref{eq:z}
  and stop.

\item Otherwise, find the rightmost letter $b$ that lies somewhere to
  the left of some letter $c$. Shift it to the right by iteratively
  applying \eqref{eq:ba} until it is immediately to the left of the
  letter $c$, then remove it using \eqref{eq:ba}.

\end{enumerate}
Repeat this until the reduction process terminates with an irreducible
word. Since $(A,\rel{S})$ is complete, reduction must terminate
with $0$. The only way a symbol $0$ can be introduced is by
application of a rule \eqref{eq:acac}, as described in case~1
above. Therefore a subword $a^{2^{k+1}-1}ca^kc$ must appear at some
point:
\begin{equation}
\label{eq:wdash}
w' \red_{\rel{S}} \alpha a^{2^{k+1}-1}ca^kc\beta \red_{\rel{S}}  0.
\end{equation}
By our choice of reduction, only case~2 is used in the reduction $w'
\red_{\rel{S}} \alpha a^{2^{k+1}-1}ca^kc\beta$. This means that this
reduction only involves shifting letters $b$ to the right using
\eqref{eq:ba} and removing letters $b$ using \eqref{eq:bc}.

Now, it is possible that $w'$ already contains a subword of the form
$a^{2^{k+1}-1}ca^kc$. Let us eliminate this case before continuing. Suppose $w' = pa^{2^{k+1}-1}ca^kcq$ for some $p,q \in
\{a,b,c\}^*$. Then
\begin{align*}
d_{\rel{R}}(w,0) &\leq d_{\rel{R}}(w,pa^{2^{n+1}-1}ca^ncq) + d_{\rel{R}}(pa^{2^{n+1}-1}ca^ncq,p0q) + d_{\rel{R}}(p0q,0) \\
&\leq |w| + |2^{n+1}-1+n| + |p|+|q|+1 \qquad\qquad \text{(by \fullref{Lemmata}{lem:distacac} \& \ref{lem:distz})} \\
&\leq |w| + |pa^{2^{n+1}-1}ca^ncq| \\
&\leq 2|w|. 
\end{align*}

So assume that $w'$ contains no such subword. Then in order for the
subword $a^{2^{n+1}-1}ca^nc$ to appear in the reduction
\eqref{eq:wdash}, the following must hold with respect to the two
distinguished letters $c$ that eventually lie in that subword: to the
left of at least one of these letters $c$ there must be some letters
$b$ that are shifted to the right using \eqref{eq:ba} and then
removing them using \eqref{eq:bc}, with the $c$ involved in this rule
being the distinguished one. Formally, there exist
$\mu_1,\mu_2\in\{a,b\}^*$ and $\alpha_1,\beta_1 \in \{a,b,c\}^*$, and
$p$ with $0\leq p\leq n$, such that
\[
w'=\alpha_1c\mu_1\mu_2c\beta_1;\qquad \alpha_1c\mu_1 \red_{\rel{S}}
\alpha a^{2^{n+1}-1}ca^p;\qquad \mu_2c\beta_1 \red_{\rel{S}}
a^{n-p}c\beta.
\]
Since no letters $b$ could be removed from $\mu_1$, it follows that
$\mu_1=a^m$ for some $m\leq p$.

Consider the reduction $\mu_2c\beta_1 \red_{\rel{S}}
a^{n-p}c\beta$. This process must remove all the letters $b$ from
$\mu_2$. There may also be some reduction involving letters from
$\beta_1$ or words to which $\beta_1$ can be reduced. But since our
rewriting system is complete, we may choose to remove all the letters
$b$ from $\mu_2$ first, and only after that to any remaining reduction
$\mu_2c\beta_1$ to $a^{n-p}c\beta$. Formally, this can be expressed as
follows. Let $l=|\mu_2|_b$. Then $l \leq n-p$ since a letter $a$ is
produced to the left of the letter $c$ whenever \eqref{eq:bc} is
applied. Then there exists $s\geq l$ such that $\mu_2c \red_{\rel{S}}
a^sca^l$ and $a^l\beta_1 \red_{\rel{S}} \beta$. For future use, let
$t',e_l,\ldots,e_1\geq 0$ be such that $\mu_2=a^{t'}ba^{e_l}\cdots
ba^{e_1}$.

Now consider the reduction $\alpha_1ca^m \red_{\rel{S}} \alpha
a^{2^{n+1}-1}ca^p$. In this reduction it is only possible to shift
letters $b$ to the right using \eqref{eq:ba} and to remove them using
\eqref{eq:bc}. To produce the correct number of letters $a$ to the
right of the distinguished letter $c$, \eqref{eq:bc} has to be applied
$p-m$ times. (We will prove later that $p-m = 0$.) The $p-m$ letters
$b$ involved in this process have to be the $p-m$ letters $b$
of the word $w'$ that are closest (to the left) to the distinguished letter
$c$. Formally, this means that there exist $t,d_{p-m},\ldots,d_1\geq
0$ such that $\alpha_1=\alpha_2 a^tba^{d_{p-m}}\cdots
ba^{d_1}ca^m$ and
\begin{equation*}
a^tba^{d_{p-m}}\cdots ba^{d_1}ca^m \red_{\rel{S}} a^{2^{n+1}-1}ca^p.
\end{equation*}

Let us summarize the information we collected so far; brackets have been added for clarity:
\begin{align}
w' ={}& \alpha_2\bigl(a^tba^{d_{p-m}}\cdots ba^{d_1}ca^m\bigr)\bigl(
a^{t'}ba^{e_l}\cdots ba^{e_1}c\bigr)\beta_1; \nonumber\\
a^tba^{d_{p-m}}\cdots ba^{d_1}ca^m \red_{\rel{S}}{} & a^{2^{n+1}-1}ca^p; \label{eq:reducetox}\\
a^{t'}ba^{e_l}\cdots ba^{e_1}c \red_{\rel{S}}{} & a^{n-p}ca^l. \label{eq:reducetoy}
\end{align}
Hence
\[
w' \red_{\rel{S}} \alpha_2\bigl(a^{2^{n+1}-1}ca^p\bigr)\bigl(a^{n-p}ca^l\bigr)\beta_1=\alpha_2\bigl(a^{2^{n+1}-1}ca^nc\bigr)\cdot
a^l\beta_1.
\]
Thus we can assume without loss of generality that $\beta=a^l\beta_1$
and $\alpha=\alpha_2$. By the definition of the function $f$, applied
to the reductions \eqref{eq:reducetox} and \eqref{eq:reducetoy},
\begin{align}
2^{n+1}-1 &= t+f(d_{p-m},\ldots,d_1), \label{eq:fappliedtox}\\
n -p &= t'+f(e_l,\ldots,e_1). \label{eq:fappliedtoy}
\end{align}
Rearranging \eqref{eq:fappliedtoy}, substituting in
\eqref{eq:fappliedtox}, and applying \fullref{Lemma}{lem:f} gives
\begin{equation}
\label{eq:derivation1}
2^{p+t'+f(\vec{e})+1}=t+2d_{p-m}+\ldots+2^{p-m}d_1+2^{p-m},
\end{equation}
where $f(\vec{e})$ is an abbreviation for $f(e_l,\ldots,e_1)$.


Now our aim is to prove that $p-m = 0$. Assume the converse, that
$p-m\geq 1$. Since $w'$ does not contain subwords $a^2b$ (recall that
we removed all such subwords in obtaining $w'$ from $w$), it follows that
\[
t\leq 1, d_{p-m}\leq 1, \ldots, d_2\leq 1.
\]
Then~\eqref{eq:derivation1} gives us
\begin{align*}
2^{p+t'+f(\vec{e})+1} &= t+2d_{p-m}+\ldots+2^{p-m}d_1+2^{p-m}\\
&\leq 1+2+\ldots+2^{p-m-1}+2^{p-m}+2^{p-m}d_1\\
&= 2^{p+1-m}+2^{p-m}d_1-1.
\end{align*}
Multiply both sides of this inequality by $2^{m-p}$ to get
\begin{equation}
\label{eq:contra1}
2^{m+t'+f(\vec{e})+1}\leq d_1+1.
\end{equation}
Now, by our chosen method of performing the
reduction~\eqref{eq:wdash}, we have that
\[
\alpha_1a^tba^{d_{p-m}}\cdots ba^{d_1}ca^m\cdot a^{t'+f(\vec{e})}ca^l\beta.
\]
is one of the intermediate words we obtain in the process of
execution. Furthermore, since we have focussed on the first subword of
the form $a^{2^{k+1}-1}ca^kc$ that appears, this intermediate word
cannot contain such a subword, and hence
\begin{equation*}
d_1<2^{m+t'+f(\vec{e})+1}-1.
\end{equation*}
This contradicts~\eqref{eq:contra1} and so $p-m=0$.

Let us pause to summarise again what we have obtained thus far:
\begin{align*}
w'={}& \alpha_2\bigl(a^tca^m\bigr)\bigl(a^{t'}ba^{e_l}\cdots ba^{e_1}c\bigr)\beta_1;\\
w' \red_{\rel{S}}{}& \alpha_2a^{t}ca^{m+t'+f(\vec{e})}ca^l\beta_1;\\
t+1 =& 2^{m+t'+f(\vec{e})+1}.
\end{align*}
Let $w'' = \alpha_2a^{t}ca^{m+t'+f(\vec{e})}ca^l\beta_1$; then $w' \red_{\rel{S}} w''$. Since
\begin{equation*}
|w'|=|\alpha_2|+t+m+t'+l+2+e_l+\ldots+e_1+|\beta_1|,
\end{equation*}
we have that
\begin{align*}
|w''| &= |\alpha_2|+t+m+t'+l+2+f(e_l,\ldots,e_1)+|\beta_1|\\
&\leq |w'|+f(e_l,\ldots,e_1)\\
&\leq |w|+2^{f(e_l,\ldots,e_1)}-1\\
&\leq |w|+t\\
&\leq |w|+|w'|\\
&\leq 2|w|.
\end{align*}
Now, by \fullref{Lemma}{lem:distacac},
\[
d_{\rel{R}}\bigl(w'',\alpha_20\beta_1\bigr) = d_{\rel{R}}\bigl(a^{t}ca^{m+t'+f(\vec{e})}ca^l,0\bigr) \leq \bigl|a^{t}ca^{m+t'+f(\vec{e})}ca^l\bigr| \leq |p| \leq 2|w|.
\]
Furthermore, by \fullref{Lemma}{lem:distz}, $d_{\rel{R}}(\alpha_20\beta_1,0)
\leq |\alpha_20\beta_1| \leq |w''| \leq 2|w|$.  Hence $d_{\rel{R}}(w'',0) \leq
4|w|$. Since every application of the rules \eqref{eq:ba} and
\eqref{eq:bc} increases the length of the word by $1$, at most $2|w|$
rewriting rules were applied in the derivation $w' \red_{\rel{S}}
w''$. Finally, the word $w'$ was obtained from $w$ by applying at most $|w|$
changes of $a^2b$ by $ba$. Therefore,
\[
d_{\rel{R}}(w,0) \leq d_{\rel{R}}(w,w') + d_{\rel{R}}(w',w'') + d_{\rel{R}}(w'',0) \leq |w| + 2|w| + 4|w| = 6|w|.
\]

Therefore, including the case where $w$ contains a symbol $0$, which
we considered previously, $d_{\rel{R}}(w,0) \leq 6|w|$.

\medskip
\noindent\textit{Second case.} Now let $w,w'\in\{a,b,c,0\}^*$ be such
that $w \thue w' \not\thue 0$. Clearly neither $w$ nor $w'$ can
contain the letter $0$. Then only the rules \eqref{eq:ba} and
\eqref{eq:bc} can be used in obtaining $w'$ from $w$. That is, $w
\thue_{\rel{T}} w'$, where $\rel{T}$ consists of the rule \eqref{eq:bc} and the rule
\begin{equation}
a^2b \to ba; \tag{AAB}\label{eq:aab}
\end{equation}
(Notice that rule \eqref{eq:aab} is rule \eqref{eq:ba} reversed.) It
will therefore suffice to prove that the Dehn function of the monoid
$N$ presented by $\pres{a,b,c}{\rel{T}}$ is linear.

Let $\rel{U}$ consist of $\rel{T}$ together with
\begin{equation}
ba^nc \to a^{2n+1}ca \qquad\text{for $n \geq 0$}.\tag{BAC}\label{eq:bac}
\end{equation}

\begin{lemma}
The rewriting systems $(\{a,b,c\},\rel{T})$ and $(\{a,b,c\},\rel{U})$
are equivalent (that is, $\thue_{\rel{T}}$ and $\thue_{\rel{U}}$
coincide) and hence $\pres{A}{\rel{U}}$ also presents the monoid $N$.
\end{lemma}

\begin{proof}
Since $\rel{T} \subseteq \rel{U}$, and the only rules in $\rel{U}$
that do not appear in $\rel{T}$ are rules \eqref{eq:bac} for $n \geq
1$, it is sufficent to observe that for any $n \geq 1$,
\begin{align*}
ba^nc \thue_{\rel{T}}{}& a^{2n}bc && \text{(by \eqref{eq:ba} applied $n$ times)} \\
\imthue_{\rel{T}}{}& a^{2n+1}ca  && \text{(by \eqref{eq:bc}).} \qedhere
\end{align*}
\end{proof}

\begin{lemma}
The rewriting system $(\{a,b,c\},\rel{U})$ is complete.
\end{lemma}

\begin{proof}
Rules \eqref{eq:bac} reduce the number of letters $b$ in a word; rules
\eqref{eq:aab} shorten a word. Hence there cannot be an infinite
sequence of applications of these rewriting rules. Hence
$(\{a,b,c\},\rel{U})$ is noetherian.

The only possible overlap of left-hand sides of rules is an overlap of
\eqref{eq:aab} with some rule \eqref{eq:bac} in a word $a^2ba^nc$:
\[
a^2ba^nc \imred_{\rel{U}} ba^{n+1}c\quad\text{ and }\quad a^2ba^nc \imred_{\rel{U}} a^{2n+3}ca;
\]
but $ba^{n+1}c \imred_{\rel{U}} a^{2n+3}ca$ by \eqref{eq:bac}. Hence
$(\{a,b,c\},\rel{U})$ is confluent.
\end{proof}

\begin{lemma}
The monoid $N$ is left-cancellative: $xu \thue_{\rel{T}} xv \implies u
\thue_{\rel{T}} v$ for all $x \in A$ and $u,v \in A^*$.
\end{lemma}

\begin{proof}
Let $u,v \in A^*$ and $x \in A$. Without loss of generality assume
that $u$ and $v$ are irreducible with respect to $\rel{U}$. Suppose
that $xu \thue_{\rel{T}} xv$. Then $xu \red_{\rel{U}} w$ and $xv
\red_{\rel{U}} w$ for some irreducible word $w$. If both $xu$ and $xv$
are irreducible, then $xu = w = xv$. So suppose, without loss of
generality, that $xu$ is reducible. Then since $u$ is irreducible, the
first application of a rewriting rule to $xu$ must include the
leftmost letter, $x$. Consider the three possibilities for $x$ in
turn:
\begin{enumerate}

\item $x = a$. Then the first rule applied must be \eqref{eq:aab} and
  so $u=abu'$. Notice that $u$ cannot contain a letter $c$, for then
  $u$ would contain $ba^kc$ as a subword for some $k \geq 0$. Let $m$
  be maximal (possibly $m=0$) such that $u' = (ab)^mu''$. If $u''$
  begins with $a$, then $u'' = a^l$ for some $l$, since $u''$ cannot
  begin $ab$ (by the maximality of $m$) or contain $a^2b$ as a
  subword. So $u = (ab)^{m+1}u''$. Thus $au = a(ab)^{m+1}u''
  \red_{\rel{U}} b^{m+1}au''$. Now either $b^{m+1}au'' =
  b^{m+1}a^{l+1}$ or $b^{m+1}au'' = b^{m+1}abu'''$, where $u'' =
  bu'''$. In either case, $b^{m+1}au''$ is irreducible. Since $av
  \red_{\rel{U}} b^{m+1}au''$, the word $av$ cannot be
  irreducible. Parallel reasoning then shows that $v = (ab)^{n+1}v''$
  and $av \red_{\rel{U}} b^{n+1}av''$, which is irreducible. Thus
  $b^{m+1}au'' = b^{n+1}av''$ and so $m=n$ and $u'' = v''$. Hence
  $u=v$.

\item $x = b$. Then the first rule applied must be \eqref{eq:bac} and
  $u = a^kcu'$ for some $k$. Then $bu = ba^kcu' \red_{\rel{U}}
  a^{2k+1}cau'$. No further rewriting can affect the prefix up to and
  including the letter $c$. So $bv$ must also rewrite to a word
  beginning with $a^{2k+1}c$. Hence $v = a^kcav'$, and $bv
  \red_{\rel{U}} a^{2k+1}cav'$. Since rewriting canot affect the prefix
    up to and including $c$, it follows that $au'$ and $av'$ must
    rewrite to the same irreducible word. Hence, since $a$
    left-cancels by part~1 above, $u' = v'$. Therefore $u=v$.

\item $x=c$. Since no rewriting rule in $\rel{U}$ has left-hand side
  beginning with $c$, the words $cu$ and $cv$ are irreducible and so
  $cu = cv$ and hence $u=v$. \qedhere

\end{enumerate}
\end{proof}

We aim to show by complete induction on $\ell=|w_1|+|w_2|$ that if
$w_1 \thue_{\rel{T}} w_2$, then $d_{\rel{T}}(w_1,w_2)\leq\ell$. The
base case of $\ell = 0$ is obvious.

Now we do the induction step. Let $w_1,w_2\in\{a,b,c\}^*$ be such that
$w_1 \thue_{\rel{T}} w_2$. By at most $\ell$ applications of the rule
\eqref{eq:aab} we perform the reductions $w_1 \red_{\rel{U}} u_1$ and
$w_2 \red_{\rel{U}} u_2$, where the words $u_1$ and $u_2$ do not
contain any subword $a^2b$. Since each application of \eqref{eq:aab}
decreases the length of the word by $1$, it follows that
$d_{\rel{T}}(w_1,u_1) \leq |w_1| - |u_1|$ and $d_{\rel{T}}(w_2,u_2)
\leq |w_2| - |u_2|$.

The number of letters $c$ in $u_1$ and $u_2$ coincide. If
there are no letters $c$ in $u_1$ or $u_2$, then $u_1$ and $u_2$ must
be irreducible and so $u_1=u_2$ (since $(A,\rel{U})$ is confluent) and
so $d_{\rel{T}}(w_1,w_2) \leq d_{\rel{T}}(w_1,u_1) +
d_{\rel{T}}(w_2,u_2) \leq |w_1| + |w_2| = \ell$. So assume that $u_1$
and $u_2$ contain at least one letter $c$. Then $u_1=p_1cq_1$ and
$u_2=p_2cq_2$ where $p_1,p_2\in\{a,b\}^*$. One easily sees that
$p_1=\emptyword$ if and only if $p_2=\emptyword$. In the case where
$p_1=p_2=\emptyword$ we have $cq_1 \thue_{\rel{T}} cq_2$, and hence
$q_1 \thue_{\rel{T}} q_2$ by left-cancellativity and then by induction
\begin{align*}
d_{\rel{T}}(w_1,w_2) &\leq d_{\rel{T}}(w_1,u_1) + d_{\rel{T}}(w_2,u_2)
+ d_{\rel{T}}(q_1,q_2) \\
&\leq |w_1| - |u_1| + |w_2| - |u_2| +|q_1|+|q_2|\\
&\leq |w_1| + |w_2| = \ell.
\end{align*}
Therefore we can assume that $p_1$ and $p_2$ are non-empty.  Again, by
left-cancellativity, we may assume that $p_1$ and $p_2$ start with
different letters. Without loss of generality, suppose $p_1=ap_1'$ and
$p_2=bp_2'$.

If $p_1'$ contains letters $b$, then since $p_1$ does not contain a
subword $a^2b$, we obtain that $p_1'=bp_1''$, where $p_1'' \in
\{a,b\}^*$. But then in reducing $p_1c = abp_1''c$ and $p_2c =
bp_2'c$ to normal form using $\rel{U}$, the last applications of
\eqref{eq:bac} to each word, both of which involve the distinguished letters $b$, produce
an odd number of letters $a$ to the left of $c$, and there is already
an extra letter $a$ in $p_1c$. Thus $p_1c=abp_1''c \red_{\rel{U}}
a^{2r}ca^d$ and $p_2c=bp_2'c \red_{\rel{U}} a^{2s+1}ca^e$, which is a
contradiction. Thus $p_1'$ cannot contain any letters $b$; hence
$p_1=a^g$.

Thus $p_1c = a^gc$ is irreducible. Recall that $p_2$ does not contain
a subword $ab$. So reducing $p_2c$ to a normal form word only
involves rules \eqref{eq:bac}, since applying such a rule cannot
create a subword $ab$ to the left of the letter $c$ and only
introduces letters $a$ to the right of the letter $c$. That is, $p_2c
\red_{\rel{U}} a^hca^k$, where $k$ is the number of letters $b$ in
$p_2$, or equivalently the number of applications of rules
\eqref{eq:bac}. Since $a^hca^k$ is irreducible, $g=h$, an since each
application of a rule \eqref{eq:bac} produces at least one letter $a$
to the left of $c$, $h\geq k$; hence $|p_1| > k$. Furthermore, $q_1
\thue_{\rel{T}} a^kq_2$.

Finally, notice that $|u_i| \geq |p_i| + |q_i| \geq k + |q_i|$ for $i=1,2$. Therefore
\begin{align*}
d_{\rel{T}}(w_1,w_2) &\leq d_{\rel{T}}(w_1,u_1) + d_{\rel{T}}(w_2,u_2)
+ d_{\rel{T}}(u_1,u_2) \\
&\leq |w_1| - |u_1| + |w_2| - |u_2| + d_{\rel{T}}(p_2c,a^hca^k) + d_{\rel{T}}(a^kq_2,q_1) \\
&\leq |w_1| - |u_1| + |w_2| - |u_2| + k + k + |q_1|+ |q_2| \\
&\leq |w_1| + |w_2| = \ell
\end{align*}

This finishes the induction step proof and we conclude that $N$ has
linear Dehn function. Hence if $w,w'\in\{a,b,c,0\}^*$ are such that $w
\thue w' \not\thue 0$, then $d_{\rel{R}}(w,w') \leq |w_1| + |w_2|$.

\medskip
This completes the proof: $M$ has linear Dehn function.
\end{proof}

\begin{proposition}
\label{prop:noregcs}
The monoid $M$ does not have a regular cross-section.
\end{proposition}

\begin{proof}
Suppose for \textit{reductio ad absurdum} that $M$ has a regular
cross-section. Then by \fullref{Proposition}{prop:changegen}, $M$
has a regular cross-section $L$ over $\{a,b,c\}$. Let $N$ be the
number of states in an automaton accepting $L$.

Choose $Q$ large enough to guarantee that
\[
(N+1)(2^{k+1}-1)\geq
2^{2^{Q+1}-1}-2 \implies k>N.
\]
Consider the word $u=a^{2^{2^{Q+1}-1}-2}ca^{2^{Q+1}-2}ca^Qc$. This
word is irreducible with respect to $\rel{S}$. Let $w$ be the unique
word in $L$ with $w \red_{\rel{S}} u$. Then since $u \not\red_{\rel{S}} 0$, the reduction $w
\red_{\rel{S}} u$ only involves rules \eqref{eq:ba} and \eqref{eq:bc}. Hence $w$
must be of the form
\[
a^tba^{d_k}\cdots ba^{d_1}ca^{t'}ba^{e_l}\cdots ba^{e_1}ca^mc,
\]
where
\begin{align*}
a^tba^{d_k}\cdots ba^{d_1}c & \red_{\rel{S}} a^{2^{2^{Q+1}-1}-2}ca^k,\\
a^{t'}ba^{e_l}\cdots ba^{e_1}ca^mc & \red_{\rel{S}} a^{2^{Q+1}-2-k}ca^Qc;
\end{align*}
hence
\begin{align*}
t + f(d_k,\ldots,d_1) &= 2^{2^{Q+1}-1}-2,\\
k+t'+ f(e_l,\ldots,e_1) &= 2^{Q+1}-2-k,\\
l+m &= Q.
\end{align*}
Assume that $d_i>N$ for some $i$. Then we can pump some power of $a$
in $a^{d_i}$ to obtain a sequence of words $w_n$ that also lie in $L$:
\[
w_n=a^tba^{d_k}\cdots ba^{d_{i+1}}ba^{d_i+qn}ba^{d_{i-1}}\cdots ba^{d_1}ca^{t'}ba^{e_l}\cdots ba^{e_1}ca^mc\in L
\]
for some $q \in \nset$. Since $f$ is strictly increasing in all of its
inputs, there is a constant $p$ such that $t +
f(d_k,\ldots,d_i+qn,\ldots,d_1) \geq 2^{2^{Q+1}-1}-1$ for all $n \geq
p$. Thus for all $n \geq p$, the word $w_n$ reduces to a word
containing $a^{2^{2^{Q+1}-1}-2}ca^{2^{Q+1}-2}c$, to which the rule
\eqref{eq:acac} applies with $n = 2^{Q+1}-2$, and so $w_n
\red_{\rel{S}} 0$. Hence there are infinitely many $w_n \in L$ with
$w_n \red_{\rel{S}} 0$; this contradicts $L$ being a cross-section.

Thus all $d_i\leq N$. Similarly $t\leq N$. Hence
\begin{align*}
2^{2^{Q+1}-1}-2 &= t+2d_k+\ldots+2^kd_1+2^k-1\\
&\leq N(1+2+\ldots+2^k)+2^k-1\\
&= N(2^{k+1}-1)+2^k-1\\
&\leq (N+1)(2^{k+1}-1),
\end{align*}
and so $k>N$ by the choice of $Q$. But then
\begin{equation*}
w \red_{\rel{S}} a^{2^{2^{Q+1}-1}-2}ca^ka^{t'}ba^{e_l}\cdots
ba^{e_1}ca^mc.
\end{equation*}
Then by pumping a power of $a$ within $a^k$ we see similarly get
infinitely many distinct words in $L$ that reduces to a words
containing $a^{2^{Q+1}-1}ca^Qc$ as a subword and hence to $0$, again
contradicting $L$ being a cross-section of $M$.
\end{proof}

\bibliography{\jobname,automaticsemigroups,languages,presentations,semigroups,c_publications}
\bibliographystyle{alphaabbrv}

\end{document}